\documentclass[aps,pra,twocolumn,showpacs]{revtex4}
\usepackage{amsmath,latexsym, amssymb, amscd,amsfonts, amsthm, epsfig,color,hyperref}
\usepackage{graphicx}
\usepackage{pictex}
\usepackage[english]{babel}

\newtheorem{theorem}{\bf Theorem}[section]

\newtheorem{definition}[theorem]{\textsl{\bf Definition}{}}
\newtheorem{lemma}[theorem]{\bf Lemma}

\begin{document}

\title{Quantum teleportation benchmarks for independent and identically-distributed \\ spin states  and displaced thermal states}
\author{M\u{a}d\u{a}lin Gu\c{t}\u{a}}\thanks{Corresponding author. \\ Electronic address: {madalin.guta@nottingham.ac.uk}}
\author{Peter Bowles}
\author{Gerardo Adesso}
\affiliation{{School of Mathematical Sciences, University of Nottingham, University Park, NG7 2RD Nottingham, United Kingdom}}
\date{September 17, 2010}


\begin{abstract}

A successful state transfer (or teleportation) experiment must perform better than the benchmark set by  the `best' measure and prepare procedure. We consider the benchmark problem for the following families of states:
 (i) displaced thermal equilibrium states of given temperature;  (ii) independent identically prepared qubits with completely unknown state.
For the first family we show that the optimal procedure is heterodyne measurement followed by the preparation of a coherent state. This procedure was known to be optimal for coherent states  and for squeezed states with the `overlap fidelity' as figure of merit. Here we prove its optimality with respect to the trace norm distance and supremum risk.
For the second problem we consider $n$ i.i.d. spin-$\frac12$ systems in an arbitrary unknown state
$\rho$ and look for the measurement-preparation pair $(M_{n},P_{n})$ for which the reconstructed state $\omega_{n}:=P_{n}\circ M_{n} (\rho^{\otimes n})$ is as close as possible to the input state, i.e. $\|\omega_{n}- \rho^{\otimes n}\|_{1}$ is small.
The figure of merit is based on the trace norm distance between input and output states.
We show that asymptotically with $n$ the this problem is equivalent to the first one.
The proof and construction of $(M_{n},P_{n})$ uses the theory of {\it local asymptotic normality} developed for state estimation which shows that i.i.d. quantum models can be approximated in a strong sense by quantum Gaussian models. The measurement part is identical with `optimal estimation', showing that `benchmarking' and estimation are closely related problems in the asymptotic set-up.
\end{abstract}

\maketitle

\section{Introduction}

Quantum teleportation \cite{Telep} and quantum state storage \cite{Storage} are by now well-established protocols in quantum information science. In both cases  the procedure amounts to mapping one quantum state onto another (at a remote location in the case of teleportation), by making use of
quantum correlations in the form of entanglement or interaction between systems. However, approximate transformations could also be accomplished without any use of quantum correlations, by means of classical `measure and prepare' (MAP) schemes.
Whilst in the ideal case, the entanglement resource gives quantum teleportation a clear advantage in terms of performance, there exists inevitable degradation of the quantum channel in realistic implementations. This has led to a number of investigations into the existence of optimal MAP schemes to locate classical-quantum boundaries and assign precise benchmarks for proving the presence of quantum effects \cite{bfkjmo}. Any experimental demonstration of quantum teleportation and state storage has to perform better than the optimal MAP scheme, to be certified as a truly quantum demonstration. A review of the quantum benchmarks for completely unknown pure input states of $d$-dimensional
systems can be found in \cite{BrussMac}. More recent research has largely focused on benchmarks originating in the context of teleportation and quantum memory for continuous variable (CV) systems \cite{COVAQIAL}, with notable results obtained for transmission of pure and mixed coherent input states, and squeezed states \cite{Hammerer,AdessoC,Namiki,Owari,Calsamiglia}. Beautiful experiments \cite{Furusawa98,memorypolzik,telepolzik,polzikalessio} involving light (Gaussian modes) and matter (coherent and spin-squeezed atomic ensembles) have demonstrated unambiguous quantum teleportation, storage and retrieval of these infinite-dimensional quantum states with a measured `fidelity' between input and output exceeding the benchmark set by the optimal MAP strategy (see also \cite{noteluk,luk}).

In each of the above cases, the benchmarks deal with the case of teleportation or storage of {\it single} input states drawn from a set, in a Bayesian or pointwise set-up. To date, there exist no nontrivial benchmarks for the transmission of multiple copies of quantum states -- a `quantum register' -- in particular for an ensemble of $n$ independent and identically-distributed (i.i.d.) qubits.
Such a task comes as a primitive in distributed quantum communication. Quantum registers can be locally initialised and then transferred to remote processing units where a quantum computation is going to take place. Also, in hybrid interfaces between light and matter \cite{qinternet}, storage and retrieval of  e.g. coherent states, involves mapping the state of $n$ i.i.d. atoms onto a light mode (back and forth). Therefore, strictly speaking, a quantum benchmark for this precise input ensemble would be needed to assess the success of the experiment. In the current practice \cite{memorypolzik,telepolzik} the problem is circumvented by noting that the collective spin components of the atomic ensemble (with $n \sim 10^{12}$ \cite{PolzikNAT}) approximately satisfy canonical commutation relations, henceforth the atomic system is treated {\it a priori} as a CV system, and the corresponding benchmarks are used.

In this paper we put this procedure on firm grounds, by proving rigorously that the optimal MAP scheme for teleportation and storage of $n$ i.i.d. unknown {\it mixed} qubits converges when $n \rightarrow \infty$ to the optimal MAP scheme  for a single-mode displaced thermal state. Additionally, we also prove that the heterodyne measurement followed by the preparation of a coherent state is optimal MAP scheme for displaced thermal
states when the figure of merit is the trace norm distance. The same scheme is known to be optimal for thermal and squeezed states, but for a figure of merit based on overlap fidelity \cite{Owari}.

The key tool in deriving our results is the theory of local asymptotic normality
(LAN) for quantum states
\cite{Guta&Kahn,Guta&Janssens&Kahn,Guta&Kahn2,Guta&Jencova} which is the quantum extension of a fundamental concept in mathematical statistics introduced by Le Cam \cite{LeCam}.  In the classical context this roughly means  that a large i.i.d. sample $X_{1},\dots , X_{n}$ from an unknown distribution contains approximately the same amount of {\it statistical} information as a single sample from a Gaussian (normal) distribution with unknown mean and known variance.


In the quantum case, LAN means that the joint (mixed)
state $\rho_{\theta}^{\otimes n}$ of $n$ identically prepared (finite dimensional) quantum systems can be transferred by means of a quantum channel to a quantum-classical Gaussian state, with asymptotically vanishing loss of statistical information. More precisely, for any fixed $\theta_{0}$ there exist quantum channels $T_{n}, S_{n}$ such that
\begin{eqnarray}
& & \quad \left\|T_{n} \left(\rho_{\theta_{0}+u/\sqrt{n}}^{\otimes n}\right) -
N_{{u}}\otimes \Phi_{{u}} \right \|_{1}  \nonumber \\
&\mbox{and}&  \quad
\left\|\rho_{\theta_{0}+u/\sqrt{n}}^{\otimes n} - S_{n} \left(
N_{{u}}\otimes \Phi_{{u}}\right) \right\|_{1} \nonumber
\end{eqnarray}
converge to zero as $n\to\infty$, {\it uniformly} over a $n^{-1/2+\epsilon}$ local neighbourhood of the state $\rho_{\theta_{0}}$. Here $N_{{u}}$ is a classical normal distribution and $\Phi_{{u}}$ is a Gaussian state on an ensemble of oscillators whose means are linear transformations of ${u}$ and the covariance matrices depend only on $\theta_{0}$. The qubit case is described in detail in Section~\ref{sec.lan} and
the precise result is formulated in Theorem \ref{th.qlan.qubits}.
%

The LAN theory has been used to find asymptotically optimal estimation procedures for qubits \cite{Guta&Janssens&Kahn} and qudits \cite{Guta&Kahn2} and to show that the Holevo bound for state estimation is achievable \cite{Guta&Kahn3}. Here we use it to solve the benchmark problem for qubits by casting it into the corresponding one for displaced thermal states. The following diagram illustrates the asymptotically optimal MAP scheme: the measurement $M_{n}$ consists in composing the channel $T_{n}$ with the heterodyne measurement $H$. The preparation procedure consists in creating the coherent state $|\alpha_{\hat{\vec{u}}}\rangle$ and mapping it back to the qubits space by the channel $S_{n}$. The optimality of the scheme is proved in Theorem \ref{th.main}.
\begin{equation}\label{comm.diagram}
\begin{CD}
\rho^{n}_{\vec{u}} @>M_n>> X_n @>P_n>>\omega(X_{n})           \\
   @V{T_n}VV @. @AA {S_n} A      \\
   \Phi_{\vec{u}} \otimes N_{\vec{u}}
   @> H >> \hat{\vec{u}} @>P>> |\alpha_{\hat{\vec{u}}}  \rangle\langle\alpha_{\hat{\vec{u}}}  |\otimes \delta_{\hat{u}_{3}}   \\
   \end{CD}\end{equation}

The paper is organised as follows.   In Section~\ref{sec.stat.formulation} we give a precise statistical formulation of the benchmark problem, and we explain in some detail the definition of the asymptotic risk. A brief overview of the necessary classical and quantum concepts from the LAN theory is given in Section~\ref{sec.lan}.
 In Section~\ref{sec.thermal} we then revisit the benchmark problem for displaced thermal states. When the figure of merit is the overlap fidelity, the solution was found in \cite{Owari} to be the heterodyne detection followed by the repreparation of a coherent state. We solve the same problem using the trace norm loss function, and again find this MAP scheme to be optimal. Interestingly, the optimality proof is based on the concept of stochastic ordering which was previously used for finding optimal cloning maps for thermal states \cite{Guta&Matsumoto}.
This result allows us to exploit LAN and solve the benchmark problem for $n$ i.i.d. qubits in the asymptotic limit $n\rightarrow\infty$. A constructive solution along with a proof of optimality is given in Section~\ref{sec.qubits}.
We conclude with discussions in Section~\ref{sec.discussion}. The Appendix contains additional mathematical details on the LAN theory for qubits.

\section{Statistical formulation of the benchmark problem}
\label{sec.stat.formulation}
The general statistical formulation of quantum benchmarking is as follows. Let
$$
\mathcal{Q}:= \{\rho_\theta: \theta\in\Theta\}
$$
be a quantum model, i.e. a family of quantum states on a Hilbert space $\mathcal{H}$, indexed by a parameter $\theta\in \Theta$. In this paper $\Theta$ is always an open subset of $\mathbb{R}^{k}$, i.e. we are in a parametric set-up.
The model encodes our prior information about the state and the parameter $\theta$ is considered to be unknown. We are given a quantum system prepared in the state $\rho_{\theta}$ and we would like to find the MAP (or classical) channel $T :\mathcal{T}_{1}(\mathcal{H})\to \mathcal{T}_{1}(\mathcal{H})$ for which $T(\rho_{\theta})$ is `close' to $\rho_\theta$. A MAP channel is the composition
$T= P\circ M$ of a measurement $M$  with outcomes in a measure space
 $(\mathcal{X},\Sigma)$, and a repreparation $P$ which assigns to every result $x\in\mathcal{X}$ a state
 $\rho_x \in \mathcal{T}_1(\mathcal{H})$. There are
several natural distance functions on the state space, such as the Bures distance \cite{Bures} with its associated Uhlmann fidelity \cite{Uhlmann}, or the trace-norm distance considered in this paper. Since $\theta$ is unknown we choose the maximum risk as the overall figure of merit of a scheme $L:=(M,P)$
\begin{equation}\label{eq.max.risk}
R_{\max}(L):= \sup_{\theta\in\Theta} \|T(\rho_{\theta}) -\rho_{\theta}\|_{1}, \qquad T=P\circ M.
\end{equation}
Alternatively one may use a Bayesian risk where the performance at different parameters is weighted by a prior distribution over $\Theta$. The goal is to find a MAP scheme with the lowest risk $R_{\max}(L)$, which will be called a {\it minmax scheme}.

The above formulation is particularly suitable in the case of covariant models such as the displaced thermal state treated in Section~\ref{sec.thermal} or the single qubit case of \cite{Calsamiglia}. Let us consider now the scenario where we want to teleport $n$ identically prepared systems, such as in the case of atomic clouds. The relevant model is then
$$
\mathcal{Q}^{n}:= \{\rho^{\otimes n}_{\theta} :\theta\in\Theta\},
$$
and we would like to find the optimal MAP scheme $L_{n}:= (M_{n},P_{n})$ for a given $n$. The experience accumulated in the related domain of state estimation indicates that this problem can rarely be solved explicitly \cite{Bagan&Gill}, but becomes tractable in an asymptotic framework
\cite{Guta&Janssens&Kahn,Guta&Kahn2}. We adopt this set-up in our benchmark problem and define the local risk of a procedure
$ L_{n}:= (M_{n},P_{n})$, around $\theta_{0}\in \Theta$ by
$$
R_{\max}(L_{n}; \theta_{0}):=
\sup_{ \| \rho_{\theta} - \rho_{\theta_{0}} \|_{1}\leq n^{-1/2+\epsilon}}
\|T_{n}(\rho_{\theta}^{\otimes n}) -\rho_{\theta}^{\otimes n}\|_{1}.
$$
where unlike  formula \eqref{eq.max.risk} we take the maximum over a neighbourhood of $\theta_{0}$ of size $n^{-1/2+\epsilon}$ with $0<\epsilon \ll 1$ arbitrary.

The {\it asymptotic local risk} of the sequence $\{L_{n}:= (M_{n},P_{n}):n\in \mathbb{N}\}$ is defined by
\begin{equation}\label{eq.local.asymptotic.risk}
R(\{L_{n} : n\in N\}; \theta_{0}):= \limsup_{n\to\infty} R_{\max}(L_{n};\theta_{0}),
\end{equation}
and by definition a MAP sequence $\{ L_{n} : n\in \mathbb{N}\}$ is optimal if it achieves the
lowest possible asymptotic local risk at any point $\theta_{0}\in \Theta$.
The latter is called the {\it asymptotic minmax risk}
$$
R_{\min\!\max}(\theta_{0})=  \limsup_{n\to\infty} \inf_{L_{n}}
R_{\max}(L_{n};\theta_{0}).
$$

Since the above formula may look rather ad-hoc to the reader who is not familiar with statistical methodology, we would like to explain its meaning in some detail.
First of all, note that we compare the input and output states globally rather than locally on each system. This means that even though as $n\to \infty$ we get more and more information about the parameter $\theta$, and we can estimate it with accuracy $O(n^{-1/2})$, in the same time the task of repreparing the state $\rho_{\theta}^{\otimes n}$ becomes more and more difficult! This can be easily understood by looking at fidelity in the case of pure states.
Let $\hat\theta_{n}$ be an estimator of
$\theta$ obtained by measuring $\psi_{\theta}^{\otimes n}$, so that
$\hat\theta_{n}-\theta =O(n^{-1/2})$, and suppose that we reprepare the state
$\psi_{\hat{\theta}_{n}}$. Then

$$
\left|\left\langle \psi_{\theta}^{\otimes n}|  \psi_{\hat{\theta}_{n}}^{\otimes n}
\right \rangle\right|^{2} =
(\cos \alpha_{n})^{2n} = 1-\alpha_{n}^{2}/2+ o(n^{-1})
$$
where $\alpha_{n}$ is an angle of order $n^{-1/2}$. Since
$$
\left(1- \frac{c}{n} +o(n^{-1})\right)^{2n} \longrightarrow \exp(-2c),
$$
we see that the input-output fidelity cannot converge to $1$. We will show that this is the case for arbitrary states, and also when we allow for other MAP schemes.

The second remark concerns the supremum over the small ball
$\|\rho_{\theta}-\rho_{\theta}\|_{1}\leq n^{-1/2+\epsilon}$ in definition
\eqref{eq.local.asymptotic.risk}. Why not consider the supremum over all $\theta$ as we did in the non-asymptotic case ? The reason is that the global supremum would be overly pessimistic and would be dominated by the region in the parameter space which is hardest for the benchmark problem. Restricting to a ball whose size is roughly that of the uncertainty in the parameter captures the local behaviour of MAP scheme at each point and is more informative than the global maximum. The ball should have size $n^{-1/2+\epsilon}$ because even if $\theta$ is unknown beforehand, it can be localised within such a region by measuring a small proportion $n^{1-\epsilon}\ll n$ of the systems, so that effectively {\it we know}
that we are in the local ball, and this should be reflected in the definition of the risk.
The localisation argument is standard in statistics and its application in quantum statistics is detailed in \cite{Guta&Janssens&Kahn}.

Finally, the relation between our figure of merit and the Bayes risk can be sketched as follows. If $R_{\pi}$ denotes the asymptotically optimal Bayes risk for the prior $\pi$,
$$
R_{\pi} := \limsup_{n\to\infty} \inf_{L_{n}} \int \pi(d\theta)
\| T_{n}(\rho_{\theta}^{\otimes n}) -\rho_{\theta}^{\otimes n}\|_{1}
$$
then under suitable conditions on $\pi$ and the model $\mathcal{Q}$ one obtains
$$
R_{\pi}= \int \pi(d\theta) R_{\min\!\max}(\theta).
$$
The intuitive explanation is that when $n\to\infty$ the features of the prior $\pi$
are washed out and the posterior distribution concentrates in a local neighbourhood of the true parameter, where the behaviour of the MAP procedures is governed by the local minmax risk.
A full proof of this relation is beyond the scope of this paper (and will be presented elsewhere) but the interested reader
may consult \cite{Belavkin&Guta} for the proof of the corresponding statement in the case of state estimation.

\section{Local asymptotic normality}
\label{sec.lan}

In this Section we will give a brief, self-contained introduction to the theory of LAN in as much detail as it is necessary for this paper and we refer to \cite{Guta&Janssens&Kahn} for proofs
and more analysis.

LAN is a fundamental concept in mathematical statistics introduced by the French statistician Le Cam \cite{LeCam}.  It roughly means  that a large i.i.d. sample $X_{1},\dots , X_{n}$ from an unknown distribution contains approximately the same amount of statistical information as a single sample from a Gaussian  distribution with unknown mean and known variance. More precisely if $X_{i}$ has distribution $\mathbb{P}_{\theta}$ depending `smoothly' on a finite dimensional parameter $\theta\in \Theta\subset\mathbb{R}^{k}$, then in a $n^{-1/2}$-size neighbourhood of any point $\theta_{0}$, the statistical model
$$
\mathcal{P}^{n}: =
\left\{
\mathbb{P}_{\theta_{0}+u/\sqrt{n}}^{n} :u\in \mathbb{R}^{k}
\right\}
$$
is well approximated by a simpler Gaussian shift model
$$
\mathcal{N}:= \left\{ N\left(u, I^{-1}(\theta_{0})\right) :u\in \mathbb{R}^{k}\right\}
$$
where $u\in\mathbb{R}^{k}$ is the local unknown parameter of the distribution, and $I^{-1}(\theta_{0})$ is the inverse Fisher information matrix at $\theta_{0}$.
Note that this approximation holds only locally, reflecting the intrinsic
uncertainty in the unknown parameter for the sample size $n$, which should not be seen as an additional assumption about $\theta$. Indeed, the parameter can always be localised  in such a region using a preliminary estimator, with vanishing probability of failure. The point of this approximation is to reduce a statistical problems (e.g. estimation) about the more complex model $\mathcal{P}^{n}$ to a simpler problem about the Gaussian model $\mathcal{N}$.

Local asymptotic normality also provides a convenient description of quantum statistical models involving i.i.d. quantum systems. Here the idea is: when the quantum `sample' is large, the model can be approximated by a simpler quantum Gaussian model. If this approximation holds in a sufficiently strong sense, then statistical problems about the qubits model can be reformulated in terms of the Gaussian one, without any loss of optimality.

In quantum statistics this technique has been used for optimal state estimation with
completely unknown finite dimensional quantum states
\cite{Guta&Kahn,Guta&Kahn2}, for optimal classification (learning)
of spin states \cite{Guta&Kotlowski}, for state transfer between matter and light
\cite{Guta&Janssens&Kahn}. In the physics literature,  LAN is used in an informal way to describe the dynamics of atomic gases in a simplified Gaussian approximation, in quantum memories and quantum metrology with spin coherent and squeezed states.

\subsection{LAN for qubit systems}

We are given $n$ independent identically prepared spin-$\frac12$ particles (qubits) in a state
$$
\rho_{\vec{r}} = \frac{1}{2}(\mathbf{1} + \vec{r}\vec{\sigma})
$$
where $\vec{r}$ is the Bloch vector of the state and
$\vec{\sigma} = (\sigma_{x}, \sigma_{y}, \sigma_{z})$ are the Pauli matrices in
$M(\mathbb{C}^{2})$.

Although a priori the state $\rho_{\vec{r}}$ is completely unknown, the following argument shows that that without loss of generality we can assume it to be
{\it localised} within a small ball of size $n^{-1/2+\epsilon}$ where $\epsilon>0$ is arbitrarily small. Indeed by measuring a small proportion $n^{1-\epsilon}\ll n$ of the systems we can devise  an initial rough estimator $\rho_{0}:=\rho_{\vec{r}_{0}}$ so that with high probability the state is in a ball of size $n^{-1/2+\epsilon}$ around  $\rho_{0}$ (see Lemma 2.1 in \cite{Guta&Kahn}). We label the states in this ball by the local parameter $\vec{u}$
$$
\rho_{\vec{u}/\sqrt{n}} =
\frac{1}{2} \left(
\mathbf{1} + (\vec{r}_{0} + \vec{u}/\sqrt{n})\vec{\sigma}\right)
$$
and define the local statistical model by
\begin{equation}\label{eq.q.n}
\mathcal{Q}_{n}:= \left\{\rho^{n}_{\vec{u}} : \| \vec{u}\|\leq n^{\epsilon} \right\} ,\qquad \rho^{n}_{\vec{u}} :=\rho^{\otimes n}_{\vec{u}/\sqrt{n}}.
\end{equation}

By choosing the reference frame with its $z$ axis along $\vec{r}_{0}$ we find that up to $O(n^{-1})$ terms the state $\rho_{\vec{u}/\sqrt{n}}$ is obtained by
perturbing the eigenvalues of $\rho_{0}$ and rotating it with a `small unitary'

$$
\rho_{\vec{u}/\sqrt{n}}
=
U_{\vec{u}/\sqrt{n}}
\left(
\begin{array}{cc}
\mu_{0}+\frac{u_{z}}{2\sqrt{n}} & 0\\
0 & 1-\mu_{0}-\frac{u_{z}}{2\sqrt{n}}
\end{array}
\right)
U_{\vec{u}/\sqrt{n}}^{\dagger},
$$
where
$$
U_{\vec{u}/\sqrt{n}}:=
\exp(i(- u_{y} \sigma_{1}+u_{x}\sigma_{2} )/2r_{0}\sqrt{n}),\qquad
r_{0}:=\|\vec{r}_{0}\|.
$$

\subsection{The big Bloch ball picture}

The asymptotic behaviour of the multiple spins state can be intuitively explained through the `big Bloch sphere' picture commonly used to describe spin
coherent \cite{Radcliffe} and spin squeezed states \cite{Kitagawa&Ueda}.
Let
$$
L_{a}:= \sum_{i=1}^{n}  \sigma^{(i)}_{a}  ,\qquad a=x,y,z
$$
be the collective spin components along the reference frame directions.
By the Central Limit Theorem, the distributions of $L_{a}$ with respect to
$\rho_{0}^{\otimes n}$ converge as
\begin{eqnarray*}
\frac{1}{\sqrt{n}}(L_{z} - n r_{0}) &\overset{\mathcal{D}}{\longrightarrow}&
N(0, 1-r_{0}^{2}),\\
\frac{1}{\sqrt{n}}L_{x,y} & \overset{\mathcal{D}}{\longrightarrow} & N(0, 1),
\end{eqnarray*}
so that the joint spins state can be pictured as a vector of length $nr_{0}$
whose tip has a Gaussian blob of size $\sqrt{n}$ representing the uncertainty in the collective variables (see Figure \ref{fig.big.ball}). Further more, by a law of large numbers argument we evaluate the commutators
\begin{eqnarray*}
&&\left[\frac{1}{\sqrt{n}} L_{x}, \frac{1}{\sqrt{n}} L_{y}\right]=
2i \frac{1}{n} L_{z} \approx
 2i r_{0} \mathbf{1},\\
&&\left[\frac{1}{\sqrt{n}} L_{x,y}, \frac{1}{\sqrt{n}} L_{z}\right]
\approx 0.
\end{eqnarray*}
Thus the rescaled observables $ L_{x}/\sqrt{2r_{0}n}$ and
$L_{y}/\sqrt{2r_{0}n}$ converge to the canonical coordinates $Q$  and $P$ of a quantum harmonic oscillator. Moreover, the variances correspond to that of a thermal equilibrium state
$$
\Phi:= (1-s) \sum_{k=0}^{\infty} s^{k} | k\rangle\langle k|,\qquad
s= \frac{1-r_{0}}{1+r_{0}},
$$
where $\{|k\rangle :k\geq 0\}$ represents the Fock basis.

\begin{figure}[t]
\includegraphics[width=7cm]{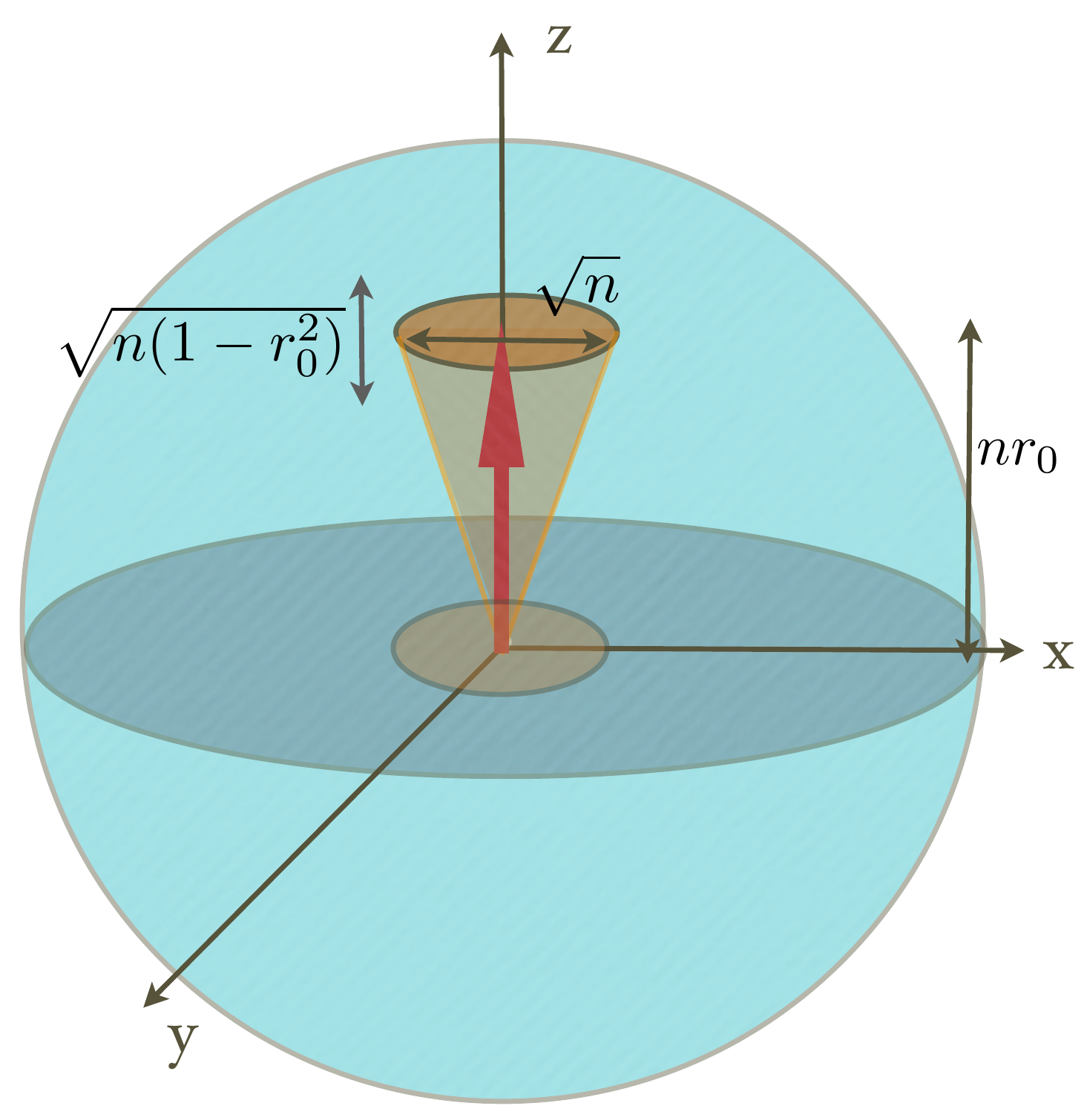}
\caption{(Color online) Big Bloch ball picture for $n$-qubit i.i.d. mixed states.}
\label{fig.big.ball}
\end{figure}

As for the third component, $(L_{z} - n r_{0})/\sqrt{n}$ converges to a
{\it classical} Gaussian variable $X\sim N:= N(0,1-r_{0}^{2})$ which is independent of the quantum state.

How does the limit change when we perturb the state of the spins ? By the same argument we find that the variables $Q,P,X$ pick up expectations which (in the first order in $n^{-1/2}$) are proportional to the local parameters $(u_{x},u_{y},u_{z})$ while the variances remain unchanged. More precisely the oscillator is in a
displaced thermal equilibrium state
$
\Phi_{\vec{u}} := D(\vec{u}) \Phi D(\vec{u})^{\dagger} ,
$
where $D(\vec{u})$ is the displacement operator
$$
D(\vec{u}):=\exp \left( i (- u_{y} Q+u_{x}P )/\sqrt{2r_{0}} \right),
$$
and the classical part has distribution $N_{\vec{u}}:=N(u_{z}, 1-r_{0}^{2})$.

We have thus identified the limit Gaussian model which is a tensor product of a classical distribution and a quantum state, which together can be seen as a state
(positive normal functional) on the von Neumann algebra
$\mathcal{B}(\ell^{2}(\mathbb{N})) \otimes L^{\infty}(\mathbb{R})$.

\begin{definition}\label{def.quantum.gaussian.shift}
The quantum Gaussian shift model $\mathcal{G}$ is defined by the family of quantum-classical states
\begin{equation}\label{eq.q.Gaussian.shift}
\mathcal{G}:= \{ \Phi_{\vec{u}}\otimes N_{\vec{u}} : \vec{u}\in\mathbb{R}^{3} \}
\end{equation}
on $\mathcal{B}(\ell^{2}(\mathbb{N})) \otimes L^{\infty}(\mathbb{R})$.
\end{definition}

In the next subsection we formulate a precise statement about the convergence to the Gaussian model which goes beyond the Central Limit type argument presented above.

\subsection{Strong convergence to Gaussian shift model}

The notion of strong convergence of classical statistical models was introduced by Le Cam and is based on defining a natural distance between statistical models with the same parameter space, so that models models at zero distance are statistically equivalent and models which are close, have similar behaviour for `regular' statistical decision problems. The existing results on quantum sufficiency \cite{Petz&Jencova} and quantum LAN \cite{Guta&Jencova,Guta&Kahn2}
indicate the existence of a theory of quantum decision and convergence of models.

In the classical set-up the distance between models is defined operationally, in
terms of randomisations which can be seen as the classical counterpart
of quantum channels.
\begin{definition}
A positive linear map
$$
T:L^{1}(\mathcal{X},\mathcal{A},\mathbb{P}) \to
L^{1}(\mathcal{Y},\mathcal{B},\mathbb{Q})
$$
is called a stochastic operator (or randomisation)
if $\|T(p)\|_{1}= \|p\|_{1}$ for every $p\in L^{1}_{+}(\mathcal{X})$.
\end{definition}
Since we work with models which may contain both classical and quantum `states' we will call a channel, a completely positive, normalised map
between preduals of von Neumann algebras. In finite dimensions, this reduces to
the familiar notion of a channel, this time with {\it block-diagonal} input and output density matrices.
\begin{definition}\label{def.quantum.LeCam.distance}
Let $\mathcal{P} := \{ \rho_{\theta}: \theta\in \Theta\}$ and
$\mathcal{Q}:= \{\sigma_{\theta}: \theta\in \Theta\} $ be two quantum statistical models over $\Theta$ with $\rho_{\theta}$ and $\sigma_{\theta}$ normal states of von Neumann algebras $\mathcal{A}$ and respectively $\mathcal{B}$.

The deficiencies
$\delta(\mathcal{P},\mathcal{Q}) $
and
$\delta(\mathcal{Q},\mathcal{P})$
are defined as
\begin{eqnarray*}
\delta(\mathcal{P},\mathcal{Q}) &:=&
\inf_{T} \sup_{\theta\in \Theta}\| T(\rho_{\theta}) -\sigma_{\theta} \|_{1}\\
\delta(\mathcal{Q},\mathcal{P}) &:=&
\inf_{S} \sup_{\theta\in \Theta}
\| S(\sigma_{\theta}) -\rho_{\theta}\|_{1}
\end{eqnarray*}
where the infimum is taken over all channels $T,S$ and $\|\cdot\|_{1}$ denotes
the $L_{1}$-norm on the preduals.

The Le Cam distance between $\mathcal{P}$ and $\mathcal{Q}$ is
$$
\Delta(\mathcal{P},\mathcal{Q}):=
{\rm max}(\delta(\mathcal{Q},\mathcal{P}) ,\, \delta(\mathcal{P} ,\mathcal{Q})).
$$
\end{definition}

With this definition we can formulate the strong convergence of the sequence
$\mathcal{Q}_{n}$ of i.i.d. qubit models to the Gaussian limit.

 \begin{theorem}\label{th.qlan.qubits}
 Let $\mathcal{Q}_{n}$ be the sequence of statistical models \eqref{eq.q.n} for $n$ i.i.d. local spin-$\frac12$ states.
and let $\mathcal{G}_{n}$ be the restriction of the Gaussian shift model
\eqref{eq.q.Gaussian.shift} to the range of parameters $\|\vec{u}\|\leq n^{\epsilon}$.
Then
$$
\lim_{n\to\infty}\Delta(\mathcal{Q}_{n}, \mathcal{G}_{n}) =0,
$$
i.e. there exist sequences of channels $T_{n}$ and $S_{n}$ such that

\begin{equation}\label{eq.channel.conv.}
\begin{split}
\lim_{n\to \infty}\,
\sup_{\| \vec{u}\|\leq n^{\epsilon}}
\| \Phi_{ \vec{u}}\otimes N_{\vec{u}}  -
T_{n} \left(  \rho_{ \vec{u}}^{n}\right)  \|_{1} =0, \\
\lim_{n\to \infty} \,
\sup_{\| \vec{u}\|\leq n^{\epsilon}} \| \rho_{\vec{u}}^{n} - S_{n} \left(  \Phi_{ \vec{u}}\otimes N_{\vec{u}} \right)  \|_{1} =0. \\
\end{split}
\end{equation}
\end{theorem}

Let us make a few comments on the significance of the above result. The first point is that LAN provides a stronger characterisation of
the `Gaussian approximation'  than the usual Central Limit Theorem arguments. Indeed the convergence in Theorem \ref{th.qlan.qubits} is strong (in $L_{1}$) rather than weak (in distribution), it is {\it uniform} over a range of local parameters rather than at a single point, and has an operational meaning based on quantum channels.

Secondly, one can exploit these features to devise asymptotically optimal measurement strategies for state estimation which can be implemented in practice by coupling with a bosonic bath and performing continuous time measurements in the bath \cite{Guta&Janssens&Kahn}.

Thirdly, the result is not restricted to state estimation but can be applied to a range of quantum statistical problems involving i.i.d. qubit states such as cloning, teleportation benchmarks, quantum learning, and can serve as a mathematical framework for analysing quantum state transfer protocols.

For completeness, we give a brief review of the main ideas involved in the proof
of Theorem \ref{th.qlan.qubits} and the description of channels $T_{n},S_{n}$ in the Appendix.

\section{Quantum benchmark for displaced thermal states with trace norm distance}
\label{sec.thermal}

In this Section we address the problem of finding the best MAP scheme for the Gaussian family of displaced thermal states with unknown mean and given variance
$$
\mathcal{T}:= \{\Phi_{z} :  z\in \mathbb{C} \}.
$$
To our knowledge this problem has only been solved in the case when the figure of merit is `overlap fidelity' \cite{Owari}. Here we show that the same procedure is optimal when the loss function is the trace norm distance.

Let $L:=(M,P)$ be a MAP procedure and define the
maximum risk as
$$
R_{\max}(L) = \sup_{z\in \mathbb{C}} \| P\circ M(\Phi_{z}) - \Phi_{z}\|_{1}.
$$

The main result  of this Section is the following.

\begin{theorem}\label{th.benchmark.Gaussian}
Let  $L^{*}:=(H,P)$ be given by heterodyne measurement followed by preparation
of a coherent state centred at the outcome of the measurement. Then $L^{*}$ is minmax i.e. for any $L$
$$
R_{\max}(L^{*})\leq  R_{\max}(L),
$$
and its risk $R^{*}(s)$ is given in Lemma \ref{lemma.risk.calculation}.
\end{theorem}

%

Before proceeding with the proof let us recall some basic definitions.
The quantum particle or `one mode' continuous variables system is characterised by the Weyl (or CCR) algebra generated by the operators $W_{\xi}$ with
$\xi\in \mathbb{C}$ satisfying the commutation relations
$$
W_{\xi} W_{\zeta}  = W_{\zeta}W_{\xi} \exp( -i {\rm Im}\langle \xi, \zeta\rangle ).
$$
The algebra is represented on the Hilbert space $L^{2}(\mathbb{R})$ as
$ W_{\xi}:=\exp(\xi a^{\dagger}-\bar{\xi} a)$ with $a,a^{\dagger}$ the creation and annihilation
operators. The latter are defined by their action on the Fock basis
$$
 |k\rangle := H_{k}(x)e^{-x^{2}/2}/\sqrt{k! 2^{k} \sqrt{\pi}} \quad k=0,1,2...
$$
such that $a^{\dagger}|k\rangle= \sqrt{k+1}|k+1\rangle$ and
$a |k\rangle= \sqrt{k}|k-1\rangle$.

A thermal state is a mixed state defined as
$$
\Phi := \sum_{k=0}^{\infty} (1-s) s^{k} |k\rangle\langle k|,
$$
where $0<s<1$ is a parameter related to the temperature by $s=e^{-\beta}$, which will be considered fixed and known. The displaced thermal states are
$\Phi_{z}:= W_{z} \Phi W_{z}^{\dagger}$. The heterodyne (or coherent) measurement is defined by its POVM
$$
H(d\zeta):= |\zeta\rangle \langle \zeta | d\zeta/2\pi
$$
where $|\zeta\rangle:=W_{\zeta}|0\rangle$ are the coherent states. For any state $\rho$ the probability density of the heterodyne outcomes is  the $Q$-function \cite{Leonhardt} $P_{\rho}(\zeta)= \langle \zeta |\rho |\zeta \rangle/ 2\pi $.
Now if we heterodyne the state $\Phi_{z}$ we obtain
$\zeta\sim N(z, \mathbf{1} (1-s)^{-1})$ and by preparing the coherent state $|\zeta\rangle$ we get the average output state
$$
T(\Phi_{z}):= P\circ H(\Phi_{z}) = \tilde{\Phi}_{z},
$$
where $\tilde{\Phi} $ is the thermal state with $\tilde{s}=(2-s)^{-1}$.
\begin{lemma}\label{lemma.risk.calculation}
The risk of the measure and prepare strategy $L^{*}:= (H,P)$ is
\begin{eqnarray*}
R_{\max}(L^{*})&:=& \sup_{z\in \mathbb{C}} \| P\circ H (\Phi_{z}) - \Phi_{z}\|_{1}\\
&=& R^{*}(s):=2 (2-s)^{-m_{0}-1}-2s^{m_{0}+1},
\end{eqnarray*}
where $m_{0}$ is the integer part of $-\log(2-s)/\log s(2-s) $.

\end{lemma}

\begin{proof}
By covariance we have $R_{\max}(L^{*})= \|\Phi -\tilde\Phi\|_{1}$. Both states are diagonal and we denote their elements $q_{i}:= (1-s)s^{i}$ and
$p_{i}:=(1-\tilde{s})\tilde{s}^{i}$ so that
$$
\|\Phi -\tilde\Phi\|_{1}=\sum_{i=0}^{\infty}| q_{i}-p_{i}|.
$$
For such geometric distributions there exists an integer
$m_{0}$ such that $p_{l} \leq q_{l}$ for $m\leq m_{0}$ and $p_{l}> q_{l}$ for
$m>m_{0}$, more precisely
$$
m_{0}
=\lfloor-\log(2-s)/\log s(2-s)\rfloor.
$$
In conclusion
$$
\| p- q\|_1 = 2\sum_{i=0}^{m_{0}}
q_{i}-p_{i}= 2 \tilde{s}^{m_{0}+1}-2s^{m_{0}+1}.
$$

\end{proof}

Figure \ref{fig.riskheterodyne} shows the decay of the risk $R(L^{*})$ as a function of the parameter $s=e^{-\beta}$ of the input state $\Phi$.
\begin{figure}[t]
\begin{center}
\includegraphics[width=8cm]{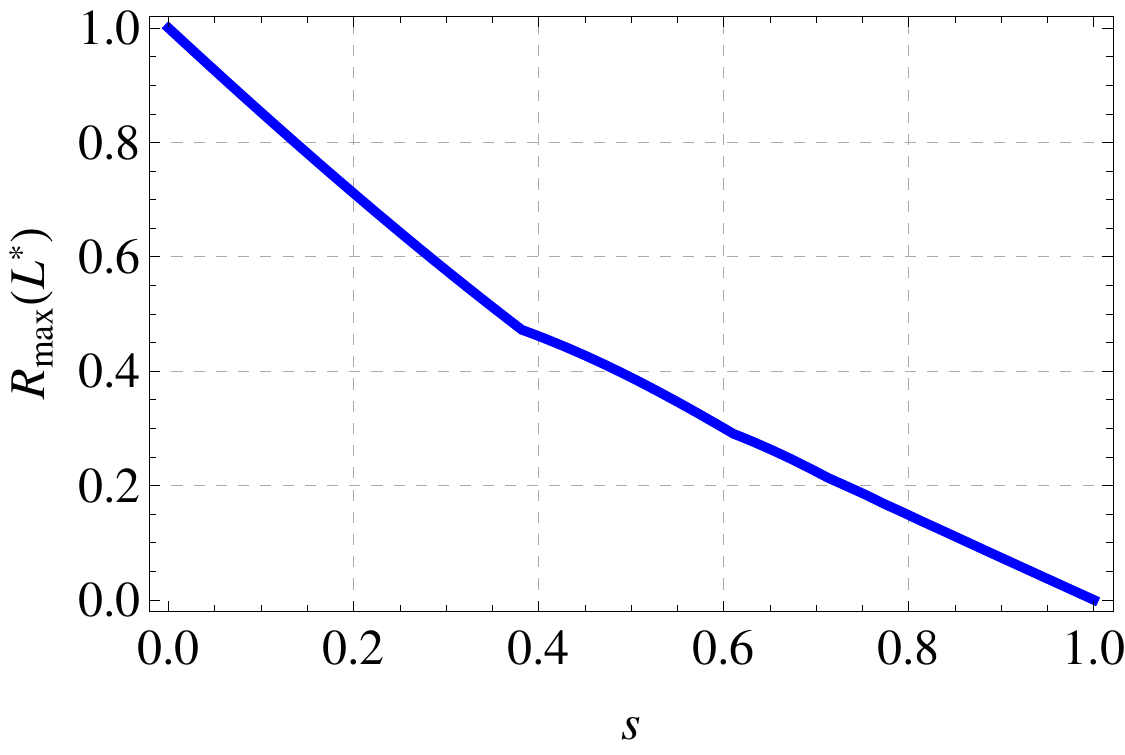}
\caption{(Color online) Risk of the optimal measure and prepare scheme as function of $s=e^{-\beta}$. All the quantities plotted are dimensionless.}
\label{fig.riskheterodyne}
\end{center}
\end{figure}

\begin{proof}
We start the proof of Theorem \ref{th.benchmark.Gaussian} by following a standard argument \cite{Owari} which shows that we can first restrict to phase space (or displacement) covariant, entanglement breaking channels, and then that it is enough to show that $L^{*}$ is optimal in a larger class of `time-reversible' channels which can be easily characterised by their action on the Weyl operators $W_{\xi}:= \exp(\xi a^{\dagger}-\bar{\xi} a)$:
\begin{equation}\label{eq.covariant.channels}
T^{\dagger}(W_{\xi})= f(\xi \sqrt{2}) W_{\xi},
\end{equation}
where $f$ is a quantum characteristic function $f(\xi)= {\rm Tr}(\tau W_{\xi})$ for some state $\tau$ with positive Wigner function, i.e. $f$ is also a classical characteristic function.

For such channels $T=P\circ M$, the risk is independent of $z$ and is equal to
$$
R_{\max}(L)= \| T(\Phi)-\Phi\|_{1}.
$$
Now, since $\Phi$ is invariant under phase rotations,
$\Phi= \exp(i\theta N)\Phi \exp(-i\theta N)$, we can apply the covariance argument again \cite{Guta&Matsumoto} to conclude that we can restrict to channels which are covariant under phase rotation, which amounts to taking $\tau$ to be diagonal in
the Fock basis. Since for diagonal states $f(\xi)=f(| \xi |)$ we obtain the following Schr\"{o}dinger version of \eqref{eq.covariant.channels}
\begin{eqnarray*}
{\rm Tr}(T(\Phi) W_{\xi})&=& {\rm Tr}(\tau W_{\sqrt{2}\xi})
{\rm Tr}(\Phi W_{\xi}) \\ &=&  {\rm Tr}(\tau W_{\sqrt{2}\xi})
{\rm Tr}(\Phi W_{-\bar{\xi}})\\
&=&
{\rm Tr}\left(\tau\otimes \Phi  \exp ( \xi c -\bar{\xi}c^{\dagger}  )\right)
\end{eqnarray*}
where
\begin{equation}\label{eq.mode.c}
c:=\sqrt{2}b+a^{\dagger}
\end{equation}
is the annihilation operator of a new mode. In other words, the output state is the same as the state of the `amplified' mode $c$ when $b$ and $a$ are prepared in states $\tau$ and respectively $\Phi$.
From this point we follow closely the arguments used in
\cite{Guta&Matsumoto} which were originally devised for finding the optimal amplification channel for displaced thermal states. The candidate channels are labelled by diagonal matrices $\tau$ and we denote by $p^{\tau}$ the probability distribution consisting of the elements of the output state
$p^{\tau}_{i}= T_{\tau}(\Phi)_{i,i}$. We further denote
$q_{i}:= \Phi_{ii}=(1-s)s^{i}$ the geometric distribution of the input thermal state and by $p^{\omega}$ the distribution of the output state corresponding to
$\omega= |0\rangle\langle 0|$. It is easy to verify that $T_{\omega}$ is the channel associated to $L^{*}:=(H,P)$, hence we would like to show that for any $\tau$
$$
\|p^{\tau}- q\|_{1} \geq \|p^{\omega} - q\|_{1}.
$$
The proof is split into two parts and relies on the notion of stochastic ordering as the
key ingredient.
\begin{definition}
Let $p=\{p_{l} : l\in \mathbb{N}\}$ and $q= \{q_{l} : l\in \mathbb{N}\}$ be two probability distributions over
$\mathbb{N}$.
We say that $p$ is stochastically smaller than $q$ ($p\preceq q$) if
$$
\sum_{l=0}^{m} p_{l} \geq \sum_{l=0}^{m} q_{l}, \quad \forall m\geq 0.
$$
\end{definition}

\begin{lemma}\label{lemma.stoch.ordering}
For any state $\tau$ the following stochastic ordering holds:
$$
p^{\omega}\preceq p^{\tau}.
$$

\end{lemma}

\begin{proof}
The first step is to `purify' $\Phi$ by writing the mode $a$ as one of the outputs of a degenerate parametric amplifier \cite{Walls&Milburn}
$$
a = \cosh(t) a_{1} + \sinh(t) a_{2}^{\dagger}
$$
with $a_{1,2}$ in the vacuum. If $\tanh (t)^{2} =s$ then the state of $a$ is $\Phi$. Plugging into \eqref{eq.mode.c} we get
$$
c =  \sinh(\tilde{t}) a_{1}^{\dagger}+ \cosh(\tilde{t}) ( T a_{2} + R b )
$$
where $ \sinh(\tilde{t})=\cosh(t), T= \sinh(t) /\cosh(\tilde{t})$ and
$R= \sqrt{2}/\cosh(\tilde{t})$ with $T^{2}+R^{2}=1$. Physically this means that the modes $a_{2}$ and $b$ are mixed with a beamsplitter and then a different degenerate parametric amplifier is applied together with mode $a_{1}$ which is in the vacuum. Note that $b$ is in the vacuum state if and only if $\tilde{b}:=T a_{2} + R b$ is also in the vacuum. For general diagonal states $\tau$ the mode $\tilde{b}$ is in the state $\tilde{\tau}$ given by the binomial formula \cite{Leonhardt}
$$
\tilde{\tau}=\sum_{k=0}^{\infty} \tau_{k} \sum_{p=0}^{k} \binom{k}{p} T^{2(p-k)} R^{2k}  |p\rangle \langle p|
= \sum_{p=0}^{\infty} \tilde{\tau}_{p}  |p\rangle \langle p| .
$$

The result of the purification argument is that we have reduced the problem of proving stochastic ordering for states of $c = a^{\dagger} + \sqrt{2}b$
when $a$ is in a {\it thermal} state, to the analogue problem for
$c= \sinh(\tilde{t}) \tilde{a}^{\dagger} +  \cosh(\tilde{t})\tilde{b}$ with $\tilde{a}:=a_{1}$ in the vacuum. Since stochastic ordering is preserved under convex combinations, it suffices to prove the statement when $\tilde\tau=|k\rangle\langle k|$ for any $k\neq 0$.

The following formula \cite{Walls&Milburn} gives a computable expression of the output two-modes vector state of the amplifier
\begin{eqnarray*}
\psi &=&
e^{\Gamma \tilde{a}^{\dagger} \tilde{b}^{\dagger}} e^{-g (\tilde{a}^{\dagger}\tilde{a}+ \tilde{b}^{\dagger}\tilde{b}+{\bf 1})}
e^{-\Gamma \tilde{a}\tilde{b}} |0,k\rangle
\\ &=&
e^{-g(k+1)} \sum_{l=0}^{\infty} \Gamma^{l} \binom{l+k}{k}^{1/2} |l,l+k\rangle,
\end{eqnarray*}
where $\Gamma = \tanh(\tilde{t})$ and $e^{g}=\cosh(\tilde{t})$. By tracing over the
 mode $\tilde{a}$ we obtain the desired state of $c$
 $$
\sum_{l=0}^{\infty} d^{k}_{l} |l \rangle \langle l|= e^{-2g(k+1)}
\sum_{l=k}^{\infty} \Gamma^{2(l-k)} \binom{l}{l-k} |l\rangle
 \langle l|.  $$
 The relation $p^{\omega}\preceq p^{\tau}$ reduces to showing that
 $$
 \sum_{l=0}^{m}  d^{0}_{l} \geq \sum_{l=0}^{m}  d^{k}_{l},
 $$
 for all $m$. If $m<k$ the right side is equal to zero and the inequality is trivial. With the notation $\gamma=\Gamma^{2}$ we get
\begin{eqnarray*}
\sum_{l=0}^{p+k} d^{k}_{l}&=&
(1-\gamma)^{k+1} \sum_{l=0}^{p}\gamma^{l} \binom{l+k}{k}\\
&=&\frac{(1-\gamma)^{k+1}}{k!} \left(\frac{1-\gamma^{k+p+1}}{1-\gamma} \right)^{(k)}\\
&=&1- \gamma^{p+1} \sum_{r=0}^{k} (1-\gamma)^{r} \gamma^{k-r} \binom{k+p+1}{r}\\
&\leq&
1-\gamma^{p+1} \sum_{r=0}^{k} (1-\gamma)^{r} \gamma^{k-r} \binom{k}{r}\\
&=&1-\gamma^{p+1} = \sum_{l=0}^{p} d^{0}_{l} \leq \sum_{l=0}^{p+k} d^{0}_{l} .
\end{eqnarray*}

\end{proof}
Stochastic ordering can be transformed into the desired optimality result by a standard argument \cite{Guta&Matsumoto} which has some interest in its own and is summarised in the following lemma. The key property needed here is that
$p^{\omega}_{l} \leq q_{l}$ if and only if $l\leq m_{0}$ (see Lemma \ref{lemma.risk.calculation}).

\begin{lemma}
The following inequality holds for any $\tau$
$$
\|p^{\tau}- q\|_{1} \geq \|p^{\omega} - q\|_{1}.
$$
\end{lemma}

\begin{proof}
Define
$$
m_{a} := \max (m: \sum_{l=0}^{m} q_{l}\leq a)
$$
and
$
\mathcal{D}(a, \tau) = \{ D\subset \mathbb{N}: \sum_{l\in D}\tau_{l}\leq a\},
$
for all $a\geq 0$. Note that by Lemma \ref{lemma.stoch.ordering} we have
$
\sum_{l=0}^{m_{a}} p^{\tau} \leq a$ for all $\tau$, and thus $\{0,1,\dots ,m_{a}\} \in \mathcal{D}(a, \tau)$.
Using the relation
$\|p-q\|_{1} = 2\sup_{D} \sum_{l\in D}(p_{l}- q_{l}) $ we obtain the chain of inequalities
\begin{eqnarray*}
\|q- p^{\tau}\|_{1} &=& 2 \sup_{a\geq 0}\, \sup_{D \in \mathcal{D}(a,\tau)}\, \sum_{l\in D}
(q_{l}- p^{\tau}_{l})\\
&
\geq&
2\sup_{a\geq 0} \sum_{l=0}^{m_{a}} (q_{l}- p^{\tau}_{l})
\geq
2\sup_{a\geq 0} \sum_{l=0}^{m_{a}} (q_{l} - p^{\omega}_{l})\\
&
=&2\sup_{m\geq 0} \sum_{l=0}^{m} (q_{l}-p^{\omega}_{l})=
\|q-p^{\omega}\|_{1}.
\end{eqnarray*}
The first equality follows directly form the definition of $\mathcal{D}(a,\tau)$.
The subsequent inequality restricts the supremum over all
$D\in\mathcal{D}(a,\tau)$ to one element $\{0,1,\dots, m_{a}\}$. In the second
inequality we replace the distribution $p^{\tau}$ by $p^{\omega}$ using the stochastic ordering proved in Lemma  \ref{lemma.stoch.ordering}.
In the subsequent equality we used the fact that both distributions $p$ and $q^{\omega}$ are geometric.

\end{proof}

We have completed the proof of Theorem \ref{th.benchmark.Gaussian}, showing that, also with respect to the trace norm figure of merit, a heterodyne measurement followed by the preparation of a coherent state centred at the measurement outcome yields the optimal MAP scheme for the transmission of displaced thermal states.
\end{proof}

\section{Asymptotic quantum benchmark for independent qubits}
\label{sec.qubits}
Here we tackle the benchmark problem for $n$ i.i.d. qubits in the asymptotic limit $n\rightarrow\infty$, we reformulate it in terms of the local coordinates according to the recipe described in Section~\ref{sec.stat.formulation}, and we obtain a constructive solution to it, by proving that this benchmark problem is equivalent to the one for displaced thermal states discussed and solved in  Section~\ref{sec.thermal}. The rigorous link between the two settings is provided by the LAN theory, whose main concepts are reviewed in Section~\ref{sec.lan}, and which will constitute an essential mathematical tool for the findings of this Section.

Given $n$ i.i.d. spin-$\frac12$ particles in state $\rho_{\vec{r}}^{\otimes n}$, we perform a measurement $M_{n}$ with outcome $X_{n}\sim \mathbb{P}^{M_{n}}_{\vec{r}}$
in a measure space $(\mathcal{X}_{n},\Sigma_{n})$ and reprepare a
quantum state $\omega(X_{n})$ on $\left(\mathbb{C}^{2}\right)^{\otimes n}$. The MAP channel is $L_{n}:=P_{n}\circ M_{n}$ and only such channels will be considered in our optimisation. A full MAP protocol will be  denoted by $L:=\{L_{n}:n\in\mathbb{N}\}$.

The risk (figure of merit) of the protocol at $\vec{r}$ for a given $n$ is
\begin{eqnarray}
R(\vec{r},L_{n})
&:=&
\| \rho^{\otimes n}_{\vec{r}} - \mathbb{E}(\omega(X_{n})) \|_{1}
 \\ &=&
\left\| \rho^{n}_{\vec{r}} -\int_{\mathcal{X}}\omega(x) \mathbb{P}^{M_{n}}_{\vec{r}} (dx)\right\|_{1} \nonumber
\end{eqnarray}

As a measure of the overall performance of $L_{n}$ we consider the {\it local maximum risk} a neighbourhood of $\vec{r}_{0}$
\begin{eqnarray}
R_{\max}(\vec{r}_{0},L_{n}) :=
\sup_{\|\vec{r}-\vec{r}_{0}\|\leq n^{-1/2+\epsilon}} R(\vec{r},L_{n})
\end{eqnarray}
whose asymptotic behaviour is
\begin{equation}
R(\vec{r}_{0}, L) :=
\underset{n\to\infty}{\lim\sup} \,
 R_{\max}(\vec{r}_{0},L_{n})
\end{equation}
\begin{definition}
The local minmax risk at $\vec{r}_{0}$ is defined by
$$
R_{\min\!\max}(\vec{r}_{0}):=\underset{n\to\infty}{\lim\sup} \, \inf_{L_{n}} R_{\max}(\vec{r}_{0},L_{n}).
$$
A protocol $L$ is called locally mimimax at $\vec{r}_{0}$ if
$
R(\vec{r}_{0}, L)=R_{\min\!\max}(\vec{r}_{0}).
$
\end{definition}

In Theorem \ref{th.main} we show that the following sequence of MAP maps is locally minmax.

\medskip

{\it Adaptive measure and prepare protocol:}

\begin{enumerate}
\item
The state is first localised in a neighbourhood
$\|\vec{r}-\vec{r}_{0}\|\leq n^{-1/2+\epsilon}$ of $\vec{r}_{0}$ by using a small proportion of the systems (see Section~\ref{sec.lan}).
\item
The remaining spins are mapped by the channel $T_{n}$ close to the classical-quantum Gaussian state  $\Phi_{\vec{u}} \otimes N_{\vec{u}}$ as in  Theorem
\ref{th.qlan.qubits}.
\item
A heterodyne measurement together with an observation of the classical
component give an estimator $\hat{\vec{u}}$ of $\vec{u}$.
\item
A coherent state $|\alpha_{\hat{\vec{u}}}\rangle$ is prepared with the mean equal to the outcome of the heterodyne measurement. The inverse channel $S_{n}$ is applied to the coherent state and the classical part.
\end{enumerate}
The procedure is illustrated in the commutative diagram \eqref{comm.diagram} drawn in the introduction: the upper line represents the MAP steps which are realised through the alternative route described in the lower part. Thus $M_{n}= H \circ T_{n}$ and
$P_{n}:= S_{n}\circ P$.


\begin{theorem}\label{th.main}
The sequence of MAP maps
$$
L^{*}_{n}:= (M_{n}:= H \circ T_{n}, \,P_{n}:= S_{n}\circ P)
$$
is locally asymptotically minmax. The minmax risk at $\vec{r}_{0}$ is equal to the benchmark for the MAP problem of a displaced thermal equilibrium state with parameter $s=(1-r_{0})/(1+r_{0})$ (see Theorem \ref{th.benchmark.Gaussian})
$$
R_{\min\!\max}(\vec{r}_{0})= R (\vec{r}_{0},L^{*})= R_{\min\!\max}(s)
$$
\end{theorem}

\begin{proof}
The idea is that the spins problem can be transferred the Gaussian one with vanishing difference in the risks. We first show that
$R_{\max} (\vec{r}_{0},L^{*})\leq R_{\min\!\max}(s)$ and then argue that a strict inequality would be in contradiction with the optimality of $R_{\min\!\max}(s)$.

By Lemma 2.1 in \cite{Guta&Janssens&Kahn} the first step of the procedure has
$o(1)$ failure probability, in which case the output state may be very different from the desired one. However since the norm distance between states is bounded by $2$, an application of the triangle inequality shows that this has no influence on the asymptotic risk. Thus we can consider that $\vec{r}$ is in the local neighbourhood
$\|\vec{r}-\vec{r}_{0}\|\leq n^{-1/2+\epsilon}$ and we can apply the
LAN machinery of Section~\ref{sec.lan}.

The (average) output state is
$$
\omega_{n}:=\mathbb{E} (\omega(X_{n}))=
S_{n} \left(
\mathbb{E} \left[
|\alpha_{\hat{\vec{u}}}  \rangle\langle\alpha_{\hat{\vec{u}}}| \right]  \otimes N_{\vec{u}} \right)
$$
where the expectation on the right side is over the heterodyne measurement outcomes. Since measurements and preparations are contractive we have
\begin{equation}\label{eq.bound1}
\|
\mathbb{E} \left[
|\alpha_{\hat{\vec{u}}}  \rangle\langle\alpha_{\hat{\vec{u}}}| \right] -
\tilde{\Phi}_{\vec{u}} \|_{1}\leq \|T_{n}(\rho^{n}_{\vec{u}}) - \Phi_{\vec{u}} \otimes N_{\vec{u}}\|_{1} =
o(1)
\end{equation}
where $\tilde{\Phi}_{\vec{u}}$ is the displaced thermal equilibrium state of variance
$1+1/2r_{0}$ obtained by heterdyning $\Phi_{\vec{u}}$, preparing the coherent state
$|\alpha_{\hat{\vec{u}}}\rangle \langle \alpha_{\hat{\vec{u}}}|$ and averaging over
$(\hat{u}_{1},\hat{u}_{2})$.

On the other hand, by using contractivity of $S_{n}$
\begin{eqnarray}\label{eq.bound2}
& &\| \rho^{n}_{\vec{u}}- \omega_{n} \|_{1} \nonumber
 \\ & & \quad \leq \| \rho^{n}_{\vec{u}}-S_{n}(\Phi_{\vec{u}} \otimes N_{\vec{u}}) \|_{1}
+
\| \Phi_{\vec{u}} - \mathbb{E} \left[
|\alpha_{\hat{\vec{u}}}  \rangle\langle\alpha_{\hat{\vec{u}}}| \right] \|_{1}\nonumber\\
& & \quad =
\| \Phi_{\vec{u}} - \mathbb{E} \left[
|\alpha_{\hat{\vec{u}}}  \rangle\langle\alpha_{\hat{\vec{u}}}| \right] \|_{1}+o(1).
\end{eqnarray}

From \eqref{eq.bound1} and \eqref{eq.bound2} we get
$$
R_{\max}(\vec{r}_{0},L^{*}) \leq \|\Phi -\tilde{\Phi}\|_{1}= R_{\min\!\max}(s).
$$
As shown in Theorem \ref{th.benchmark.Gaussian}, the right side is the MAP benchmark for mixed Gaussian states with optimal procedure consisting of heterodyne measurement followed by the preparation of a coherent state with mean equal to the outcome of the measurement.

Now we show that no MAP strategy can achieve a lower asymptotic risk at $\vec{r}_{0}$ than  $R_{\min\!\max}(s)$. Indeed suppose that
$\tilde{L}:= (\tilde{M}_{n},\tilde{P}_{n})$ is a sequence of procedures for spins  which satisfies $R_{\max}(\vec{r}_{0},\tilde{L}) < R_{\min\!\max}(s)$. Then as shown below, we could construct a MAP procedure for $\Phi_{\vec{u}}$ which is strictly better that the optimal one, which is impossible.

Let $\delta>0$ be a small number to be fixed later. We mix $\Phi_{\vec{u}}$ with the thermal state of the same temperature $\Phi_{0}$, through a beamsplitter with small reflectivity $r=\delta$ and transmissivity $t= \sqrt{1-\delta^{2}}$ and obtain the output
$$
\Phi_{t \vec{u}}\otimes \Phi_{r\vec{u}}.
$$
By heterodyning $\Phi_{r \vec{u}}$ we obtain an estimator $\vec{u}_{0}$ such that
$$
\mathbb{P} \left[ \|\vec{u} - \vec{u}_{0}\| \geq L \right] \leq \epsilon_{2}/2
$$
for some (large) $L$ which increases when $\delta,\epsilon_{2}\downarrow 0$.

We displace the unmeasured component
$\Phi_{t \vec{u}}$ by $-t\vec{u}_{0}$ so that from now on we can assume that
$\|\vec{u}\|\leq L$ (with an $\epsilon_{2}$ loss of risk).

Now we choose
$n$ large enough so that $L \leq n^{\epsilon}$. Thermal states with such displacements are in the range of applicability of the inverse map $S_{n}$ in LAN Theorem \ref{th.qlan.qubits}.  We apply the channel $S_{n}$ mapping the Gaussian state to an i.i.d. spins state, with the small difference that the classical parameter is now fixed to zero, i.e. the spins will be prepared in a state with Bloch vector of length $r_{0}$. By Theorem \ref{th.qlan.qubits} we have uniformly in $\|\vec{u}\|<L$
$$
\|S_{n}(\Phi_{t\vec{u}}\otimes N_{0}) - \rho^{n}_{t\vec{u}}\|_{1} = o(1).
$$

Next we measure and re-prepare the spins using the procedure $\tilde{L}$ and obtain  a state $\omega_{n}:= \mathbb{E}\left[\omega(X_{n})\right]$ such that
$$
\|\omega_{n} - \rho^{n}_{t\vec{u}}\|_{1} \leq
R_{\max}(\vec{r}_{0},\tilde{L}) +o(1).
$$

Finally, we apply the map $T_{n}$ to the output state $\omega(X_{n})$ and keep only the quantum part $T^{(q)}(\omega_{n})$. By the same contractivity argument as before we have
\begin{eqnarray}
\|\Phi_{t\vec{u}} -  T^{(q)}_{n} (\omega_{n})\|_{1} &\leq&
\|\Phi_{t\vec{u}} - T^{(q)}_{n} (\rho^{n}_{t\vec{u}})\|_1+
\|\omega_{n} - \rho^{n}_{t\vec{u}}\|_{1}\nonumber\\
&\leq& R_{\max}(\vec{r}_{0},\tilde{L}) + o(1).
\end{eqnarray}

At this point we can directly compare out state $T^{(q)}_{n} (\omega_{n})$ with the target $\Phi_{\vec{u}}$, or use a quantum amplifier to make up for the loss in amplitude induced by the initial use of a beam splitter. We follow the second  `unbiased' line. The amplifier can be described by a linear transformation on the mode $a$ of the oscillator together with an ancillary mode $b$ prepared in the vacuum. The output modes are
\begin{eqnarray*}
c_{1}&:=& \sqrt{t^{-1}} a+ \sqrt{t^{-1}-1} b^{\dagger}\\
c_{2}&:=&\sqrt{ t^{-1}} b+ \sqrt{t^{-1}-1} a^{\dagger}.
\end{eqnarray*}
Let $A$ denote the channel mapping the state of the mode $a$ to that of the amplified mode $c_{1}$. Then
$$
A\Phi_{t\vec{u}}= \Phi^{\prime}_{\vec{u}}
$$ where $\Phi^{\prime}$ is a thermal equilibrium state with variance
$$
\tilde{V}= t^{-2}/2r_{0} +(t^{-1}-1)/2.
$$

The final distance estimate (conditional on successful
localisation of $\Phi_{\vec{u}}$) is
\begin{eqnarray}
\| \Phi_{u} - A\circ T^{(q)}_{n}(\omega_{n})\|_{1} \!\!&\leq&\!
\|\Phi_{\vec{u}}- \Phi^{\prime}_{\vec{u}}\|_{1} +
\| \Phi_{t\vec{u}} - T_{n}(\omega_{n})\|_{1}\nonumber\\
\!\!&\leq&\!\|\Phi-\Phi^{\prime}\|_{1} + R_{\max}(\vec{r}_{0},\tilde{L}) +o(1). \nonumber \\ & &
\end{eqnarray}

Since we assumed that $R_{\max}(\vec{r}_{0},\tilde{L})< R_{\min\!\max}(s)$, it is enough to choose $\epsilon_{2}$ and $\|\Phi-\Phi^{\prime}\|_{1}$ small enough to obtain a contradiction with the optimality of $R_{\min\!\max}(s)$.  This can be done by choosing $\delta, \epsilon_{2}$ small enough and $n$ large enough.

\end{proof}

\section{Discussion and concluding remarks}
\label{sec.discussion}

The problem of finding an optimal MAP reconstruction scheme for a family of quantum states has attracted significant attention due to its relevance
in establishing fidelity benchmarks for teleportation and state storage. In the case of a family of displaced thermal (and coherent) states with unknown displacement the problem was solved in  \cite{Hammerer,Owari} and the optimal procedure is the heterodyne measurement followed by the preparation of a coherent state.
We showed that the same MAP procedure is again optimal for this family of states, with a more natural figure of merit - the trace norm distance. Moreover, in the case of i.i.d. mixed qubit states, the benchmark problem can be solved asymptotically by mapping it to the previous problem, using LAN theory.

Interestingly, the same heterodyne measurement is also optimal from the point of view of state estimation \cite{Holevo}. On the other hand, in \cite{Calsamiglia} it was shown that for a particular family of states in $\mathbb{C}^{2}$, the optimal MAP scheme involves a measurement which is different from the optimal one for state estimation. One of our motivations was to see whether this peculiarity survives in the asymptotic limit, and the conclusion is that for large $n$ the estimation and benchmarking can be performed optimally {\it simultaneously}, i.e. their measurement parts are identical.


Another difference between benchmarking and estimation pointed out in
\cite{Calsamiglia} refers to the preparation part of the protocols. More exactly, it turns out that in the case of a special family of one qubit pure states, the optimal re-prepared state {\it does not} belong to the family. However, our asymptotic benchmark for qubits can be easily extended to treat the case of pure rather than mixed states and the result is that  {\it asymptotically}, the optimal re-prepared state {\it is} in the original model. Thus the effect pointed out in \cite{Calsamiglia} is due to the particular geometry of the states space in the one sample situation, which `linearises' asymptotically. To briefly explain our claim, note that asymptotically, the parameter space reduces effectively to a small interval and the tangent space approximation kicks in. Then the problem is transferred to that of finding the benchmark for a family of coherent states {\it on a line}, where the solution is a homodyne measurement followed by the preparation of a coherent state {\it in the family}.

Oh the other hand, we have seen that in the case of mixed states, the re-prepared state is pure, hence not in the family, even in the asymptotic framework.  However this fact is not surprising as it happens already in the classical case: if $X$ is a random variable with distribution $\mathbb{P}$ then the best `re-preparation' is the `pure state' represented by the $\delta$ measure
$\delta_{X}$ leading to a trivial benchmark. If $\mathbb{P}$ is not a $\delta$ measure itself then the re-prepared `state' $\delta_{X}$ is outside the model.


\acknowledgments{
M\u{a}d\u{a}lin Gu\c{t}\u{a} was supported by the EPSRC Fellowship EP/E052290/1. We thank A. Serafini for useful discussions.}

\bigskip

\appendix*

\setcounter{equation}{0}

\section{Explicit construction of the channels $T_{n},S_{n}$}\label{secApp}

We give here a brief review of the main ideas involved in the proof
of Theorem \ref{th.qlan.qubits} and the description of channels $T_{n},S_{n}$.

On $\left(\mathbb{C}^{2}\right)^{\otimes n}$ we have two {\it commuting}
unitary group representations
\begin{eqnarray*}
\pi_{n}(U)&:&\psi_{1}\otimes \dots\otimes \psi_{n} \mapsto
U \psi_{1}\otimes \dots \otimes U\psi_{n} ,\\
\tilde{\pi}(t) &:&  \psi_{1}\otimes \dots\otimes \psi_{n}\mapsto
\psi_{t^{-1}(1)}\otimes \psi_{t^{-1}(n)}
\end{eqnarray*}
where $U\in SU(2)$ and $t\in S(n)$, with $S(n)$ denoting the symmetric group. By Weyl's Theorem, the representation space decomposes into a direct sum of tensor products
\begin{equation}\label{eq.decomposition}
\left( \mathbb{C}^{2}\right)^{\otimes n} = \bigoplus_{j=0, 1/2}^{n/2} \mathcal{H}_{j} \otimes \mathcal{H}^{j}_{n},
\end{equation}
where $j$ is half-integer,  $\mathcal{H}_{j} \cong \mathbb{C}^{2j+1}$ is an irreducible representation of $SU(2)$, and $\mathcal{H}^{j}_{n}\cong \mathbb{C}^{n_{j} }$ is the irreducible representation of $S(n)$.
Since the density matrix $\rho^{n}_{\vec{u}}$ is invariant under
permutations, it has a block diagonal form
\begin{equation}
\label{blocks}
\rho^{n}_{\vec{u}} = \bigoplus_{j=0, 1}^{n}  p^{n}_{\vec{u}}(j) \rho^{n}_{j,\vec{u}} \otimes
\frac{\mathbf{1}}{n_{j}} .
\end{equation}

This can be interpreted as being given a random variable $J$ with distribution
$p^{n}_{\vec{u}}(j)$ and conditionally on $J=j$, a quantum state
$\rho^{n}_{j,\vec{u}}$. The classical and quantum components of
the `data' will be `processed' by randomising $J$ and
mapping $\rho^{n}_{j,\vec{u}}$ through a quantum channel.

{\it Classical component.}--- Each block is an eigenspace of  the total spin operator
$L^{2}= L_{x}^{2}+L_{y}^{2}+L_{z}^{2}$ with eigenvalue $4j(j+1)\approx (2j)^{2}$.
As the big Bloch ball argument  suggests, the main contribution to $L^{2}$
comes from $L_{z}^{2}$ so that the distribution $p^{n}_{\vec{u}}(j)$ can be approximated by the binomial
$$
p^{n}_{\vec{u}}(j) \approx  {n\choose j+n/2}
\left(\frac{1+r}{2}\right)^{j+n/2}  \left(\frac{1-r}{2}\right)^{n/2-j}
$$
with $r=r_{0}+u_{z}$. The fact that $p^{n}_{\vec{u}}(j)$ converges to
$N_{\vec{u}}$ follows now from the classical version of LAN \cite{LeCam} for i.i.d. samples of binary variables. We first constructs the rescaled variable
$$
G_{n}:=\sqrt{n} \left(\frac{L}{n} - r_{0} \right)
\overset{\mathcal{D}}{\longrightarrow} N(u_{z}, 1-r_{0}),
$$
but since $G_{n}$ has a discrete probability distribution, we need to `smooth' it by randomising with e.g. a Gaussian Markov kernel of variance
$1/(2\sqrt{n})$
$$
K_{n,j}(x) :=  (n^{1/4}/\sqrt{\pi})\exp\left(-\sqrt{n} (x-G_n(j))^2\right).
$$

In this way the convergence in distribution is converted into strong
($L_{1}$) convergence.

{\it Quantum component.}---
Conditionally on obtaining $J=j$ in the which-block measurement, we are left with a quantum state $\rho^{n}_{j,\vec{u}}$ on $\mathcal{H}_{j}$.
The action of $T_{n}$ will be to imbed it into the quantum oscillator
space by the isometry $V_{j}:\mathcal{H}_{j}\to \ell^{2}(\mathbb{Z}) $ define below.

Let $\pi_{j}$ denote the irreducible representation of $SU(2)$ on
$\mathcal{H}_{j}$ and denote  its generators by $L_{j,a}= \pi_{j}(\sigma_{a})$.
The space has an orthonormal
basis $\left\{|j,m\rangle , m=-2j, \dots, 2j\right\}$  such that
$$
L_{j,z}  |j,m \rangle  = m |j,m\rangle.
$$

The Quantum Central Limit Theorem suggests that the properly normalised operators $L_{j, \pm} := L_{j,x} \pm i L_{j,y}$ converge to the annihilation and creation operators $a,a^{\dagger}$ when $j\approx n r_{0}/2\to\infty$. Since they act as ladder operators
\begin{eqnarray*}
&&
J_{j, +} |j,m\rangle = \sqrt{j-m} \sqrt{ j+m+1} \, |j,m+1\rangle ,\\
&&
J_{j, -} |j,m\rangle = \sqrt{j-m+1} \sqrt{ j+m}  \,  |j,m-1 \rangle .
\end{eqnarray*}
it is natural to define the isometry
$$
V_{j}:|m,j\rangle \mapsto |j-m\rangle
$$
defining the embedding (channel)
$$
T_{j}: \rho^{n}_{j,\vec{u}}\longmapsto V_{j}  \rho^{n}_{j,\vec{u}} V_{j}^{j}.
$$

Putting everything together we obtain the channel $T_{n}$
$$
T_{n} : \rho^{n}_{\vec{u}}
\longmapsto
\sum_{j} p^{n}_{\vec{u}} (j) K_{n,j} \otimes V_{j}  \rho^{n}_{j,\vec{u}} V_{j}^{j}
$$
which implements the convergence to Gaussian in Theorem \ref{th.qlan.qubits}.

The channel $S_{n}$ is basically an inverse of $T_{n}$: the normal distribution is discretised to produce the distribution $p^{n}_{\vec{j}}$ and the quantum Gaussian state is compressed to the first $2j+1$ levels and mapped to $\mathcal{H}_{j}$ with the co-isometry $V_{j}^{\dagger}$. For more details on the proof we refer to \cite{Guta&Janssens&Kahn}.



\begin{thebibliography}{33}
\bibitem{noteluk}
Let us mention, however, that an experimental implementation of the conventional benchmark criteria for CV systems can be highly demanding. Alternative approaches to the derivation of benchmarks for quantum channels, such as memory
or teleportation channels, have been then pursued \cite{luk}. These benchmarks, using entanglement-verification tools such as the `expectation value matrix', directly assess the ability of quantum transmission channels (as realized in an experiment) to preserve entangled states, as opposed to MAP channels which are entanglement-breaking. Such testing procedures, encompassing also the realistic possibility of using mixed states to probe the quantum channels, might be verified with very few experimental resources.
\end{thebibliography}

\end{document}